\newcommand{\FullOrShort}{short}   %

\RequirePackage{ifthen}
\ifthenelse{\equal{\FullOrShort}{full}}{
		\documentclass[11pt]{article}
	  \newcommand{\fullOnly}[1]{#1}
	  \newcommand{\shortOnly}[1]{}
  }{
		\documentclass[11pt]{article}
	  \newcommand{\fullOnly}[1]{}
	  \newcommand{\shortOnly}[1]{#1}
  }

\usepackage[utf8]{inputenc}
\usepackage{mathtools,amssymb,mathrsfs,amsthm}
\usepackage[normalem]{ulem}
\usepackage{float}
\usepackage{framed}
\usepackage{subfigure}
\usepackage[table,svgnames]{xcolor}
\usepackage[colorlinks=true,linkcolor=black,citecolor=MidnightBlue,urlcolor=MidnightBlue]{hyperref}
\usepackage{xspace}
\usepackage{tikz}
\usetikzlibrary{automata,patterns,calc}
\usepackage{pgfplots}
\usetikzlibrary{positioning}
\usepackage[hmargin=1in,vmargin=1in]{geometry}
\usepackage{microtype}
\usetikzlibrary{trees}
\usepackage{booktabs}
\usepackage{multirow}
\usepackage{stackrel}
\usepackage[margin=2em,font=small,labelfont=bf]{caption}

\usepackage{algorithm}
\usepackage{algorithmicx}
\usepackage[noend]{algpseudocode}

\usepackage{lmodern}
\usepackage[T1]{fontenc} %
\usepackage{comment}

\newtheorem{theorem}{Theorem}[section]

\newtheorem{claim}[theorem]{Claim}
\newtheorem{definition}[theorem]{Definition}
\newtheorem{lemma}[theorem]{Lemma}

\newtheorem{corollary}[theorem]{Corollary}

\newcommand{\eps}{\varepsilon}

\newcommand{\ch}{\mathsf{CH}}

\newcommand{\CC}{\textnormal{CC}}

\renewcommand{\phi}{\varphi}
\newcommand{\owner}{\mathsf{owner}}

\newcommand{\goodrate}{constant rate\xspace}     %

\newcommand{\bscf}{{feedback channel}\xspace}
\newcommand{\bec}{{erasure channel}\xspace}

\newcommand{\maj}{\mathsf{maj}}

\newcommand{\back}{\gets}

\newcommand{\conf}{\textit{conf}}
\newcommand{\Trim}{\textsf{Trim}}

\usepackage{svn}
\usepackage{svn-multi}

\svnidlong
{$HeadURL$} {$LastChangedDate$} {$LastChangedRevision: 55 $} {$LastChangedBy$}

\begin{document}

\title{\textbf{Maximal Noise in Interactive Communication \\ over Erasure Channels and Channels with Feedback}}

\author{
Klim Efremenko \\[0.2ex] \small{UC Berkeley\thanks{Work done while at Univ.\@ of Chicago.}} \\[-0.4ex] \small{\texttt{klimefrem@gmail.com}}
\and
Ran Gelles \\[0.2ex] \small{Princeton University\thanks{Work done while a student at UCLA.}}  \\[-0.4ex] \small{\texttt{rgelles@cs.princeton.edu}}
\and 
Bernhard Haeupler \\[0.2ex] \small{Carnegie Mellon University}  \\[-0.4ex] \small{\texttt{haeupler@cs.cmu.edu}}
}

\date{}

\maketitle

\thispagestyle{empty}
\setcounter{page}{0}

\begin{abstract}
We provide tight upper and lower bounds on the noise resilience of interactive communication
over noisy channels with \emph{feedback}. 
In this setting, we show that the maximal fraction of noise that any robust protocol can resist is~$1/3$. 
Additionally, we provide a simple and efficient robust protocol that succeeds as long as the fraction of noise is at most~$1/3-\varepsilon$.
Surprisingly, both bounds hold regardless of whether the parties 
send bits or symbols from an arbitrarily large alphabet.

We also consider interactive communication over \emph{erasure} channels. We provide a protocol that matches the optimal tolerable erasure rate of $1/2-\varepsilon$ of previous protocols (Franklin et~al., CRYPTO~'13) but operates in a much simpler and more efficient way. Our protocol works with an alphabet of size~$4$, in contrast to prior protocols in which the alphabet size grows as~$\varepsilon\to0$. 
Building on the above algorithm with a \emph{fixed} alphabet size, 
we are able to devise a protocol for \emph{binary} erasure channels
that tolerates erasure rates of up to~$1/3-\varepsilon$.
\end{abstract}

\clearpage

\section{Introduction}
In the interactive communication setting, Alice and Bob are given 
inputs $x$ and~$y$ respectively, and are required to compute and output some
function~$f(x,y)$ of their joint inputs. 
To this end, they exchange messages over a channel that may be noisy:
up to an $\eps$-fraction of the transmitted bits may get flipped during the communication. 
Due to the noise, there is a need for error correction and sophisticated coding schemes
that will allow the parties to successfully conduct the computation, 
yet keep the communication complexity small, ideally at most linear in the communication complexity of computing the same function 
over a noiseless channel (hereinafter, we say that such a scheme has a \emph{\goodrate}).

Coding schemes for interactive communication have been extensively explored, starting with the pioneering work of Schulman~\cite{schulman92,schulman93,schulman96}
who gave the first \goodrate scheme to resist up to a $1/240$-fraction of bit flips. Almost two decades later, Braverman and Rao~\cite{BR11}  showed a \goodrate coding scheme that successfully computes any function, as long as the fraction of corrupted transmissions is at most~$1/4-\eps$. Furthermore, they show that it is impossible to resist noise of~$1/4$ or more, for a large and natural class of \emph{robust} protocols. 
In robust protocols both parties are guaranteed to agree  whose turn it is to speak at each round, regardless of the noise, e.g., when their order of speaking is a fixed function of the round number (see definition in Section~\ref{sec:prelim} below). It should be noted that the above result of~$1/4-\eps$ applies only when the parties send symbols from a large alphabet, whose size is growing as $\eps$ goes to zero. When the parties are restricted to sending bits, the coding scheme of Braverman and Rao~\cite{BR11} tolerates up to a $(1/8-\eps)$-fraction of bit flips. The question of determining the maximal tolerable noise for binary channels is still open.

In this paper we examine  different types of communication channels and noise.
Specifically, we consider  \emph{channels with feedback} and \emph{erasure channels}.
In the former it is assumed that after each transmission, the sender learns the (possibly corrupted) symbol received by the other side, i.e., there is a noiseless feedback. In the erasure channel case, the noise can turn any symbol into an ``erasure'' (denoted~$\bot$), but it cannot alter the transmission into a different valid symbol. Both erasure channels and channels with feedback have been studied in the classical one-way setting \cite{shannon48,berlekamp1964block} albeit typically more from a perspective of optimizing communication rates.  

For each of these channels we examine the maximal tolerable noise for interactive communication,
both when the parties are restricted to sending bits and in the general case where they are allowed to send symbols from a larger alphabet.

\subsection{Our Results}
\begin{table}[htb]
\renewcommand{\arraystretch}{1.2} 
\centering
\begin{tabular}{@{}lllcc@{}}
\toprule
\textbf{channel type} & \textbf{alphabet} & \textbf{order of speaking} & \textbf{lower bound} & \textbf{upper bound} \\
\midrule
feedback & ternary \& large  & fixed & 1/4 & 1/4 \\
feedback & binary & fixed & 1/6 & 1/6 \\
feedback & binary \& large & arbitrary & 1/3 & 1/3 \\
\hline
erasure & 4-ary \& large & fixed & 1/2 & 1/2 \\
erasure & binary & fixed & 1/3 & ??  \\
\bottomrule
\end{tabular}
\caption{A summary of the lower (achievability) and upper (impossibility) bounds for the maximum tolerable error rate for all settings considered in this paper.}
\label{table:results}
\end{table}

\paragraph{Interactive communication over channels with feedback.}

We completely solve the question of the maximal tolerable noise for robust interactive protocols over channels with feedback, both for the binary alphabet and the large alphabet case. 
We note that while in the standard noisy model the parties in a robust protocol must have a fixed order of speaking which depends only on the round of the protocol~\cite{BR11}, this is not the case for noisy channels with feedback.
Indeed, due to the feedback both parties know the symbols received at the other side and may determine the next party to speak according to their joint view. While this decision may depend on the noise, the parties maintain  a consensus regarding the next party to speak. We can therefore refine the class of robust protocols into ones in which the order of speaking is fixed (i.e.,  a function of the round number) and ones in which the order is arbitrary (i.e., possibly dependent on the noise). We stress that these two subsets of protocols are still robust, and refer the reader to~\cite{GHS14,AGS13} for a discussion on \emph{adaptive} (non-robust) protocols.

As a helpful warm-up we first consider protocols with a fixed order of speaking. When the parties are allowed to send symbols from a large alphabet, we show for any $\eps>0$ an efficient coding scheme with a \goodrate that resists a noise rate of up to~$1/4-\eps$. 
Although the same bounds were already given by~\cite{BR11,GH14} for standard noisy channels, our protocol is considerably simpler while also being computationally efficient. 
Moreover, while in other schemes the size of the alphabet increases as $\eps \to 0$, in our protocol a \emph{ternary} alphabet suffices. %
The main idea is the following: the parties exchange messages as in the noiseless protocol, and use the feedback to verify that the messages were received intact. In case of a corrupted transmission, the parties transmit a special symbol~`$\back$' that instructs both parties to rewind the protocol to the step before the corrupted transmission. Building on the above coding scheme we provide for any~$\eps>0$ a simple and efficient 
\emph{binary} protocol  that resists up to a $(1/6-\eps)$-fraction of bit flips. 
\begin{theorem}
For any $\eps>0$ and any function~$f(x,y)$ there exists an efficient robust coding scheme \emph{with a fixed order of speaking} and a \goodrate
that correctly computes~$f(x,y)$ for each of the following settings:
(i) over a channel with feedback with ternary alphabet, assuming at most a $(1/4-\eps)$-fraction of the symbols are corrupted, (ii) over a binary channel with feedback, assuming at most a $(1/6-\eps)$-fraction of the bits are corrupted
\end{theorem}

Additionally, we prove that the above bounds of $1/4$ and $1/6$ are tight for the general feedback channel, and the binary feedback channel, respectively. 
The impossibility result for the binary case has a particular interesting implication:
since  feedback channels are more powerful than standard noisy channels, 
this impossibility applies also to robust protocols over standard binary noisy channels (i.e., without a feedback), 
narrowing the maximal tolerable noise for this setting to the region~$[1/8, 1/6]$.
\begin{theorem}\label{thm:lowerbound1/6}
There exists a function $f(x,y)$,
such that any robust binary interactive protocol,
succeeds in computing $f(x,y)$ with probability at most $1/2$ assuming a $1/6$-fraction of bit-flips.
\end{theorem}

Next, we consider robust protocols with arbitrary (noise-dependent) order of speaking.
In this case the simple idea presented above immediately gives a higher noise-resilience of~$1/3$. The reason for this discrepancy in the bounds when we allow the order of speaking to be arbitrary stems from the following issue. 
When a transmission is corrupted, the \emph{sender} is aware of this event and it sends a rewind symbol~`$\back$' on the next time it has the right to speaks. 
However, when the order of speaking is fixed (say, alternating), the parties ``lose'' one slot: 
while we would like the sender to repeat the transmission that was corrupted, the \emph{receiver} is the next party to speak after the round where the `$\back$'~symbol is sent.
If we allow the order of speaking to be arbitrary, we can avoid this excessive round and thus improve the noise resilience.

Translating the above idea to the binary case gives a protocol that resists a noise rate of~$1/5-\eps$. However we can do better---we devise a protocol that resists noise rates up to~$1/3-\eps$.
Here the parties send  messages of varying length, consisting of the original information followed by a varying amount of \emph{confirmation bits}. The confirmation bits indicate whether or not the information was corrupted by the adversary. This practically forces the adversary to spend more of its corruption budget per message, or otherwise the receiving party learns about the corruption and simply ignores the message. 
\begin{theorem}
For any $\eps>0$ and any function~$f(x,y)$ there exists an efficient robust coding scheme with \goodrate
that correctly computes~$f(x,y)$ over a binary channel with feedback assuming at most a $(1/3-\eps)$-fraction of the bits are corrupted.
\end{theorem}
It is interesting to mention that in contrast to all the previous settings 
and in contrast to the case of standard (uni-directional) error correction, 
the size of the alphabet (binary or large) makes no difference to the noise resilience of this setting.

We also provide a matching impossibility bound of~$1/3$ that applies to any alphabet size, 
and in particular to the binary case. 
\begin{theorem}
There exists a function~$f(x,y)$, 
such that any robust interactive protocol over a channel with feedback (with any alphabet) 
that computes~$f(x,y)$,
succeeds with probability at most~$1/2$ if a $1/3$-fraction of the transmissions are corrupted.
\end{theorem}

\paragraph{Interactive communication over erasure channels.}
In \cite{FGOS13} it was shown that the maximal noise over erasure channels when
a large alphabet can be used is $1/2-\eps$. This is trivially tight for protocols with a fixed order by completely erasing all the symbols sent by the party that speaks less. [In fact, this applies to any \emph{robust} protocol---we show that robust protocols over erasure channels must have a fixed order of speaking!]
When the parties are restricted to using a binary alphabet, 
it is possible to resist an erasure rate of $1/4-\eps$~\cite{FGOS13,BR11}.
The main drawback of these coding schemes is that they are not computationally efficient for the case of adversarial noise, and can take exponential time to complete in the worst case.

Here we suggest a coding scheme with a \goodrate that can tolerate an erasure rate of up to~$1/2-\eps$,
yet it is computationally efficient and very simple to implement. Moreover, our ``large'' alphabet is of size~$6$, regardless of~$\eps$.
\begin{theorem}
For any $\eps>0$ and any function $f(x,y)$ there exists an efficient, robust coding scheme with \goodrate
that correctly computes $f(x,y)$ over an erasure channel with a $6$-ary alphabet, assuming at most a $(1/2-\eps)$-fraction of the bits are corrupted.
\end{theorem}
The approach here is slightly different than the above schemes over channels with feedback. We no longer use a $\back$ symbol to rewind the simulation in a case of error, but instead each party always sends the next message according to the transcript accepted by that party so far. The key point is that erasures cannot make the receiver accept a wrong message. In other words, if the receiver receives a symbol ($\ne\bot$) it is guaranteed that this is indeed the symbol that was sent by the sender. It follows that the only issue possibly caused by an erasure is getting the players \emph{out of sync}, i.e., getting them to simulate different rounds of the protocol. However, we show that this discrepancy in the simulated round is limited by $\pm1$. Thus, sending a small parity of the party's current simulated round (e.g., its round number modulus~$3$) is enough to re-gain synchronization and proceed with the simulation.  

\medskip
Interestingly, the small and fixed alphabet size serves as a stepping stone in devising 
a protocol that works for \emph{binary} erasure channels.
Encoding each symbol of the $6$-ary alphabet in the above scheme 
into a binary string yields a protocol that resists 
erasures fractions of up to~$3/10-\eps$. Yet, we are able to optimize the above simulation and reduce the size of the alphabet to only~$4$ symbols. This allows us to encode each symbol in the alphabet using a binary code with a very high relative distance, and obtain a protocol that 
tolerates a noise rate of~$1/3-\eps$. 
This improves over the more natural and previously best known bound of $1/4 - \eps$. 
\begin{theorem}
For any $\eps>0$ and any function $f(x,y)$ there exists an efficient, robust coding scheme with \goodrate
that correctly computes $f(x,y)$ over a binary erasure channel, assuming at most a $(1/3-\eps)$-fraction of the bits are corrupted.
\end{theorem}

The only impossibility bound we are aware of is again the trivial bound of~$1/2$~\cite{FGOS13,GH14} which applies even to larger alphabets. We leave determining the optimal erasure rate for coding schemes over binary erasure channels as an interesting open question.

\medskip
We summarize our results in Table~\ref{table:results}.

\subsection{Other Related Work}

\paragraph{Maximal noise in interactive communication.}
As mentioned above, the question of interactive communication over a noisy channel
was initiated by Schulman~\cite{schulman92,schulman93,schulman96} who mainly focused on the case of random bit flips, but also showed that his scheme resists an adversarial noise rate of up to~$1/240$. Braverman and Rao~\cite{BR11} proved that $1/4$ is a tight bound on the noise (for large alphabets),
and Braverman and Efremenko~\cite{BE14} gave a refinement of this bound, 
looking at the noise rate separately at each direction of the channel (i.e., from Alice to Bob and from Bob to Alice). For each pair of noise rates, they determine whether or not a coding scheme with a \goodrate exists.
Another line of work improved the efficiency of coding schemes for the interactive setting, 
either for random noise~\cite{GMS11,GMS14}, or for adversarial noise~\cite{BK12,BN13,GH14}. 

Protocols in the above works are all robust. The discussion about non-robust or \emph{adaptive} protocols was initiated by Ghaffari, Haeupler and Sudan~\cite{GHS14,GH14} and concurrently by Agrawal, Gelles and Sahai~\cite{AGS13}, giving various notions of adaptive protocols and analyzing their noise resilience. Both the adaptive notion of~\cite{GHS14,GH14} and of~\cite{AGS13} are capable of resisting a  higher amount of noise than the maximal~$1/4$ allowed for robust protocols. Specifically, a tight bound of~$2/7$ was shown in~\cite{GHS14,GH14} for protocols of fixed length; when the length of the protocol may adaptively change as well, a coding scheme that achieves a noise rate of~$1/3$ is given in~\cite{AGS13},  yet that scheme does not have a \goodrate.

\paragraph{Interactive communication over channels with feedback and erasure channels.}
To the best of our knowledge, no prior work considers the maximal noise of interactive communication over noisy channels with feedback. Yet, within this setting, the maximal \emph{rate} of coding schemes, i.e., the minimal communication complexity as a function of the error rate, was considered by~\cite{Pankratov13,GH15} (the rate of coding schemes in the standard noisy channel setting was  considered by~\cite{KR13,Haeupler14}).

For erasure channels, a tight bound of~$1/2$ on the erasure rate of robust protocols was given in~\cite{FGOS13}. For the case of adaptive protocols,~\cite{AGS13} provided a coding scheme with a \goodrate that resists a relative erasure rate of up to~$1-\eps$ in a setting that allows parties to remain silent in an adaptive way. 
The case where the parties share a memoryless erasure channel with a noiseless feedback was considered by Schulman~\cite{schulman96} who showed
that for any function $f$, the communication complexity of solving~$f$ in that setting equals the distributional complexity of~$f$ (over noiseless channels), up to a factor of the channel's capacity.

\section{Preliminaries}
\label{sec:prelim}

We begin by setting some notations and definitions we use throughout.
We sometimes refer to a bitstring $a \in \{0,1\}^n$ as an array $a[0],\dotsc, a[n-1]$.
$a \circ b$ denotes the concatenation of the strings $a$ and~$b$.
$\mathsf{prefix}_k(a)$ denotes the first $k$~characters in a string~$a$, and
$\mathsf{suffix}_k(a)$ denotes the last $k$~characters in~$a$.
For two strings $a,b$ of the same length~$n$, their Hamming distance $d(a,b)$ is the number of indices $0\le i \le n-1$ for which $a[i]\ne b[i]$.

\begin{definition}\label{def:channels}
A \emph{feedback channel} is a channel $\ch:\Sigma\to\Sigma$ in which at any instantiation 
noise can alter any input symbol~$\sigma\in\Sigma$ into any output~$\sigma'\in\Sigma$. 
The sender is assumed to learn the (possibly corrupt) output $\sigma'$ via a noiseless feedback channel.

An \emph{erasure channel} is a channel $\ch: \Sigma \to \Sigma \cup \{\bot\}$ in which the channel's noise is restricted into changing the input symbol into an erasure symbol~$\bot$.

For both types of channels, the \emph{noise rate} is defined as the fraction of corrupted transmissions out of all the channel instantiations.
\end{definition}

An interactive protocol $\pi$ over a channel~$\ch$, is a pair of algorithm
$\pi_{Alice}, \pi_{Bob}$ that determine the next message to be communicated, given the input and the transcript so far. 
The communicated message is assumed to be a single symbol from the channel's alphabet~$\Sigma$. In all our protocols, $|\Sigma|=O(1)$. The protocol runs for $|\pi|$ rounds after which each party computes an output as a function of its input and the transcript that party sees. The Protocol is said to compute a function~$f(x,y)$ if for any pair of inputs $x,y$, both parties output~$f(x,y)$.  
A coding scheme $\Pi$ is said to \emph{simulate} $\pi$ %
if for any pair of inputs~$x,y$, the parties output $\pi(x,y)$---the transcript of running $\pi$ on input~$(x,y)$ over a noiseless channel

We further assume that at every given round only one party sends a message. Protocols in which the identity of the sender of each round is well defined and agreed upon both parties are called robust.
\begin{definition}\label{def:robust}
We say that an interactive protocol $\pi$ is \emph{robust} \textnormal{(\cite{BR11})} if, \\ 
(1) for all inputs, the protocol runs for $n$~rounds; 
(2) at any round, and given any possible noise, 
the parties are in agreement regarding the next party to speak. 

A \emph{fixed order} protocol is one in which condition (2) above is replaced with the following \\
(2') 
there exist some function $g: \mathbb{N} \to \{\text{Alice}, \text{Bob}\}$ such that
at any round $i$, the party that speaks is determined by $g(i)$, specifically, it is independent of the noise.
\end{definition}
Note that any fixed-order protocol is robust, but it is possible that a robust protocol will not have a fixed order (see, e.g., Algorithm~\ref{alg:bscf-binary}, or the coding scheme of Theorem~\ref{thm:bscf-adp-large} in the case of channels with noiseless feedback.).
In that case we say that the robust protocol has an \emph{arbitrary} or \emph{noise dependent} order of speaking. 

In the following we show how to take any binary alternating (noiseless) protocol, and simulate it over a noisy channel. Note that for any function $f$ there exists a binary alternating (noiseless) protocol~$\pi$, such that the communication of~$\pi$ is linear in the communication complexity of~$f$, that is, 
\[
\CC(\pi) = O(\CC(f)).
\]
Hence, simulating the above $\pi$ with communication $O(\CC(\pi))$ has a \goodrate, since its communication is linear in $\CC(f)$.

\section{Feedback Channels with a Large Alphabet:  Upper and Lower Bounds} 

\subsection{Protocols with a fixed order of speaking}
Let us consider simulation protocols in which the order of speaking is fixed and independent of the inputs 
the parties hold, and the noise injected by the adversarial channel. 
We show that~$1/4$ is a tight bound on the tolerable noise in this case.
The bound is the same as in the case of standard noisy channels (without feedback)~\cite{BR11}.
We begin with the upper bound, by showing a protocol that correctly simulates $\pi$ assuming  noise level of~$1/4-\eps$. It is interesting to note that the alphabet used by the simulation protocol is independent of~$\eps$ (cf.~\cite{BR11,FGOS13,GH14,BE14}); specifically, we use a \emph{ternary alphabet}. In addition the simulation is deterministic and (computationally) efficient, given black-box access to $\pi$.

\begin{theorem}\label{thm:bscf-fixed-large}
For any alternating noiseless binary protocol $\pi$ of length~$n$, and for any~$\eps>0$,
there exists an efficient, deterministic, robust simulation of~$\pi$ over 
a \bscf using an alphabet of size~3 and a fixed order of speaking, that takes $O_\eps(n)$~rounds and
succeeds assuming a maximal noise rate of~$1/4-\eps$.
\end{theorem}
\begin{proof}
We use a ternary alphabet~$\Sigma=\{0,1,\back\}$. 
The simulation works in alternating rounds where the parties run~$\pi$, and verify via the feedback that any transmitted bit is correctly received at the other side. Specifically, if the received symbol is either a~$0$ or a~$1$ the party considers this transmission as the next message of~$\pi$, and extends the simulated transcript~$T$ accordingly. If the received symbol is~$\back$, 
the party rewinds three rounds of~$\pi$, that is, the party deletes the last four undeleted symbols of~$T$.\footnote{The four symbols removed from~$T$ are the received `$\back$' symbol plus three simulated rounds of~$\pi$.}
Each party, using the feedback, is capable of seeing whether the transcript~$T$ held by the other side contains any errors, and if so, it sends multiple~$\back$ symbols until the corrupted suffix is removed. 
The above is repeated $N = n/4\eps$ times (where $n=|\pi|$), and at the end the parties output~$T$.
The protocol is formalized in Algorithm~\ref{alg:feedback-fixed-large}.
\enlargethispage{1em}

\begin{algorithm}[hbt!]
\caption{A fixed-order simulation for channels with feedback}
\label{alg:feedback-fixed-large}
\begin{algorithmic}[1]
\small
\Statex
	Input: a binary alternating protocol~$\pi$ of length $n$, a noise parameter~$\eps>0$, an input value~$x$.
\Statex

\Statex
Assume a fixed alternating order of speaking: Alice is the sender on odd $i$'s,  and Bob is the sender on even $i$'s.

\Statex 
\State
	 Set $N = \lceil n/4\eps \rceil$, initialize $T \gets \emptyset$; $T^F \gets \emptyset$.
\Statex

\Comment $T$~is the simulated transcript as viewed by the party. We can split~$T$ into two substrings corresponding to alternating indices: $T^S$ are the sent characters, and $T^R$ the received characters. Let~$T^F$ be the characters received by the other side (as learned via the feedback channel).
\Statex
	
\Statex

\For {$i = 1$ to $N$} \label{step:for}
	\If {$T^F = T^S$}
	\Statex
	\Comment run one step of $\pi$, given the transcript so far is~$T$

			\State $T \gets T \circ \pi(x \mid T)$
			\State $T^F \gets T^F \circ \langle\text{symbol recorded at the other side}\rangle$
	\Else
 			\State if sender:  \\
				\qquad\qquad send~a `$\back$' symbol \\
				\qquad\qquad $T \gets T\circ{}$`$\back$'\\
				\qquad\qquad $T^F \gets T^F \circ \langle\text{symbol recorded at the other side}\rangle$.
			\State  if receiver: \\ 
				\qquad\qquad extend $T$ according to incoming symbol.  

	\EndIf
	\Statex
	\If {$\mathsf{suffix}_1(T^R)=$`$\back$'\  \ or \ $\mathsf{suffix}_1(T^F)=$`$\back$'}
		\State $T \gets \mathsf{prefix}_{|T|-4}(T)$ \label{step:last} 
		\Statex 	 \quad\qquad (also delete the corresponding transmissions in~$T^F$)
	\EndIf
\EndFor
\State Output~$T$
\end{algorithmic}
\end{algorithm}

Note that due to the alternating nature of the simulation, 
each corruption causes four rounds in which~$T$ doesn't extend:
(1) the corrupted slot; (2) the other party talks; (3) sending a~$\back$ symbol; (4) the other party talks. After step (4) the simulated transcript~$T$ is exactly the same as it was before (1). Also note that consecutive errors (targeting the same party\footnote{consecutive corruptions targeting the other party will be corrected without causing any further delay.}) simply increase the amount of $\back$~symbols the sender should send, so that each additional corruption extends the recovery process by at most another four rounds. 
Also note that corrupting a bit into a~$\back$ has a similar effect: after four rounds, $T$~is back to what it was before the corruption: (1) the corrupted slot; (2--4) re-simulating $\pi$ after three bits of~$T$ were deleted.

Therefore, with $1/4-\eps$ noise, we have at most $4 \cdot {(1/4-\eps)}N = N(1-4\eps)$ rounds that are used to recover from errors and do not advance~$T$.
Yet, during the rest $4\eps N = n$~rounds $T$~extends correctly and the simulation succeeds to output the entire transcript of~$\pi$.
\end{proof}

Next we prove it is impossible to tolerate noise rates above~$1/4$.
\begin{theorem}\label{thm:fixedorderlb}
Any protocol with a fixed order of speaking that computes the identity function $f(x,y)=(x,y)$,
succeeds with probability at most~$1/2$ over a~\bscf assuming $1/4$ of the transmission are corrupted. 
\end{theorem}
\begin{proof}
The proof is similar to the case of interactive communication over of a 
standard noisy channel (without feedback)~\cite{BR11}.
Assume that Alice speaks for $R$~rounds and without loss of generality assume~$R\le N/2$ 
(note that since the protocol has a fixed order of speaking, 
the party that speaks in less than half the rounds is independent of the input and noise, and is well defined  at the beginning of the simulation). 
Define $\mathsf{EXP0}$ to be an instance in which Alice holds the input~$x=0$, and 
we corrupt the first $R/2$~rounds in which Alice talks 
so that they are the same as what Alice would have sent had she held the input~$x=1$. 
Define $\mathsf{EXP1}$ to be an instance in which Alice holds the input~$x=1$, and 
we corrupt the last $R/2$~rounds in which Alice talks so that they are the same as what Alice sends 
during the same rounds in~$\mathsf{EXP0}$.

Note that from Bob's point of view (including his feedback) $\mathsf{EXP0}$ and $\mathsf{EXP1}$ are indistinguishable, thus Bob cannot output the correct $x$ with probability higher than $1/2$. In each experiment
we corrupt only half of Alice's slots, thus the total noise is at most~$R/2 \le N/4$.
\end{proof}

\subsection{Protocols with an arbitrary order}

It is rather clear that the protocol of Theorem~\ref{thm:bscf-fixed-large} ``wastes'' one round (per corruption) only due to the fixed-order of speaking: when a corruption is noticed and a~$\back$ symbol is sent, the parties would have liked to rewind only \emph{two} rounds of $\pi$, exactly back to beginning of the round that was corrupted. 
However, this will change the order of speaking, since that round belongs to the same party that sends the~$\back$ symbol. This suggests that if we lift the requirement of a fixed-order simulation, and allow the protocol to adapt the order of speaking, the simulation will resist up to a fraction $1/3$ of noise. In the following we prove that $1/3$ is a tight bound on the noise for this case. 

We remark that although the protocol is adaptive in the sense that the order of speaking may change due to the noise, both parties are always in consensus regarding who is the next party to speak. Indeed, using the feedback channel, both parties learn the symbols \emph{received} at both sides. 
Such a joint view can uniquely determine the next party to speak,
thus, the protocol is robust (Definition~\ref{def:robust}).

\begin{theorem}\label{thm:bscf-adp-large}
For any alternating noiseless binary protocol $\pi$ of length~$n$, and for any $\eps>0$,
there exists an efficient, deterministic, robust simulation of~$\pi$ over 
a \bscf using an alphabet of size~3, that takes $O_\eps(n)$ rounds and
succeeds assuming a maximal noise rate of $1/3-\eps$.
\end{theorem}
\begin{proof}
We use a ternary alphabet~$\Sigma=\{0,1,\back\}$. The simulation protocol is similar to 
Algorithm~\ref{alg:feedback-fixed-large}:
each party maintains a simulated transcript~$T$, and uses the feedback to verify that the other party holds a correct simulated transcript. As long as there is no noise in the simulated transcript,  the parties continue to simulate the next step of~$\pi$ given that the transcript so far is~$T$. Otherwise, the party that notices a corruption sends a~$\back$ symbol at the next round assigned to that party.
When a~$\back$ symbol is received, the party rewinds~$\pi$ by \emph{two} rounds, that is, the party deletes the last three symbols of~$T$.\footnote{The three symbols removed from~$T$ are the received `$\back$' symbol plus two rounds of~$\pi$.} 
The next party to speak is determined by $\pi(x \mid T^R, T^F)$; note that $(T^F,T^R)_{\text{Alice}} = (T^R,T^F)_{\text{Bob}}$, thus the parties progress according to the same view and are in-synch at all times.
The above is repeated $N = n/3\eps$ times (where $n=|\pi|$), and at the end the parties output~$T$.

It is easy to verify that each corruption causes at most three recovery rounds, after which $T$~is restored to its state prior the corruption: 
(1) the corrupted slot; (2) the other party talks; (3) sending a~$\back$ symbol; After step (3) the simulated transcript~$T$ is exactly the same as it was before (1), and the party that spoke at (1) has the right to speak again.
Again, note that consecutive errors simply increase the amount of $\back$ symbols the sender should send, so that each additional corruption extends the recovery process by at most another three rounds. 
Also note that corrupting a bit into a~$\back$ has a similar effect: after three rounds, $T$~is back to what it was before the corruption: (1) the corrupted slot; (2--3) re-simulating the two bits of~$T$ that were deleted.

When the noise level is bounded by~$1/3-\eps$, 
we have at most $3 \cdot (1/3-\eps)N = N(1-3\eps)$~rounds that are used to recover from errors 
and do not advance~$T$; yet, during the rest $3\eps N = n$~rounds $T$~extends correctly.
Therefore, at the end of the simulation the parties output a transcript of~$\pi$ with a correct prefix of length at least~$n$, thus they successfully simulate~$\pi$.
\end{proof}

\begin{theorem}\label{thm:imp-bscf-adp}
Any robust protocol that computes the identity function $f(x,y)=(x,y)$ over a~\bscf with an error rate of~$1/3$, succeeds with probability at most~$1/2$.
\end{theorem}
\begin{proof}
The proof is based on ideas from~\cite{GHS14} for proving a lower bound on the noise tolerable by adaptive protocols over a standard noisy channel (without feedback).

Consider a protocol of length~$N$, and suppose that on inputs~$x=y=0$, Bob is the party that speaks less during the first $2N/3$~rounds of the protocol. Recall that due to the feedback, we can assume the parties are always in consensus regarding the party to speak on the next round, so that at every round only a single party talks; thus Bob talks at most $N/3$ times during the first $2N/3$~rounds.
Consider the following experiment $\mathsf{EXP1}$ in which $x=0,y=1$ however we corrupt Bob's messages during the first $2N/3$ rounds so that they are the same as Bob's messages given $y=0$. 
From Alice's point of view, the case where Bob holds $y=0$ and the case where $y=1$ but all his messages are corrupted to be as if he had $y=0$, are equivalent. Therefore, with the consensus assumption, in both cases Bob's talking slots are exactly the same, and this strategy corrupts at most $N/3$~messages.

Now consider the following experiment $\mathsf{EXP0}$ in which $x=y=0$, however, during the last $N/3$ rounds of the protocol we corrupt all Bob's messages to be the same as what he sends in $\mathsf{EXP1}$ during the same rounds. Note that due to the adaptiveness of the order of speaking in the protocol, it may be that Bob talks in \emph{all} these $N/3$~rounds, but corrupting all of them is still within the corruption budget.

Finally, in both $\mathsf{EXP0}$ and $\mathsf{EXP1}$ Alice's view (messages sent, received and feedback) is the same, implying she cannot output the correct answer with probability higher than~$1/2$.
\end{proof}

\section{Feedback Channels with a Binary Alphabet:  Upper and Lower Bounds}
We now turn to examine the case of  feedback channels with a binary alphabet. 
We begin (Section~\ref{sec:bscf-bin-fixed}) with the case that the robust simulation has a fixed order of speaking, and show a tight bound of~$1/6$ on the noise.
We then relax the fixed-order requirement (Section~\ref{sec:bscf-bin-arbitrary})
and show that $1/3$ is a tight bound on the noise in this case. It is rather surprising that simulations with binary alphabet reach the same noise tolerance 
of~$1/3$ as simulations with large alphabet.

\subsection{Protocols with a fixed order of speaking}
\label{sec:bscf-bin-fixed}

\begin{theorem}\label{thm:bscf-fixed-binary}
For any alternating noiseless binary protocol~$\pi$ of length~$n$, and for any~$\eps>0$,
there exists an efficient, deterministic, robust simulation of~$\pi$ over 
a \bscf with a binary alphabet and a fixed order of speaking, that takes $O_\eps(n)$~rounds and
succeeds assuming a maximal noise rate of~$1/6-\eps$.
\end{theorem}
\begin{proof}
In Algorithm~\ref{alg:feedback-fixed-large} we used a special symbol~$\back$ to signal 
that a transmission was corrupted and instruct the simulation to rewind. When the alphabet is binary such a symbol can not be used directly, but we can code it into a binary string, e.g.,~``$00$''. To this end, we first need to make sure that the simulation does not communicate $00$ unless a transmission got corrupted.
However, recall that while no corruption is detected, Algorithm~\ref{alg:feedback-fixed-large} simply communicates the transcript of the protocol~$\pi$ it simulates. We therefore we need to preprocess~$\pi$ so that no party sends two consecutive zeros (cf.~\cite{GH15}). 
This can easily be done by padding each two consecutive rounds in~$\pi$ by two void rounds where each party sends a~`$1$' (two rounds are needed to keep the padded protocol alternating). Denote with~$\pi'$ the preprocessed protocol, and note that~$|\pi'|=2|\pi|$. 

We now simulate~$\pi'$ in a manner similar to Algorithm~\ref{alg:feedback-fixed-large}.
The parties communicate in alternating rounds where at each round they send the next bit defined by~$\pi'$ according to the current simulated transcript~$T$.
In case that a corruption is detected via the feedback, the sender
sends the string~$00$ to indicate a rewind request. Due the alternating nature of the simulation, it takes three rounds to complete communicating the rewind request. Whenever a party receives a~$00$ rewind command, both parties delete the last~$6$ bits of~$T$ (recall that both parties know the received symbols: one directly and the other via the feedback). Observe that we must remove an even number of bits so that the alternating order of speaking is maintained. Thus, although the erroneous bit is only 5 rounds prior to receiving the~$\back$ command, we rewind six rounds, see Figure~\ref{fig:6bits}. The simulation is performed for~$N=|\pi'|/6\eps = O(|\pi|)$ rounds at the end of which both parties output~$T$.

The analysis is similar to the proof of Theorem~\ref{thm:bscf-fixed-large}, and we thus omit here the full details. The only difference is that each corruption takes at most six rounds to recover.
This implies a maximal tolerable noise level of~$1/6-\eps$. 
\end{proof}

\begin{figure}[htb]
\centering
\begin{framed}
\begin{tikzpicture}
\foreach \x/\y in {-1.5/Alice,0/1, 1/0,2/1,3/0,4/1,5/0}
	\node at (\x,1.2) {\y};
	
\foreach \x/\y in {-1.5/Bob,0/1, 1/\textcolor{red}{\textbf{1}},2/1,3/0,4/1,5/0}
	\node at (\x,0) {\y};

\foreach \x/\y in {1/red,3/,5/}
	\node at (\x,0.6) [\y] {$\downarrow$};	
\foreach \x/\y in {0,2,4}
	\node at (\x,0.6) {$\uparrow$};

\draw[->] (5.2,-0.2) -- (5.2,-0.7) --  (-0.25,-0.7) node[below,midway] {rewind} -- (-.25,-0.2);
	
\end{tikzpicture}
\end{framed}
\caption{Illustration of rewinding the protocol after the first bit sent by Alice is corrupted.}
\label{fig:6bits}
\end{figure}
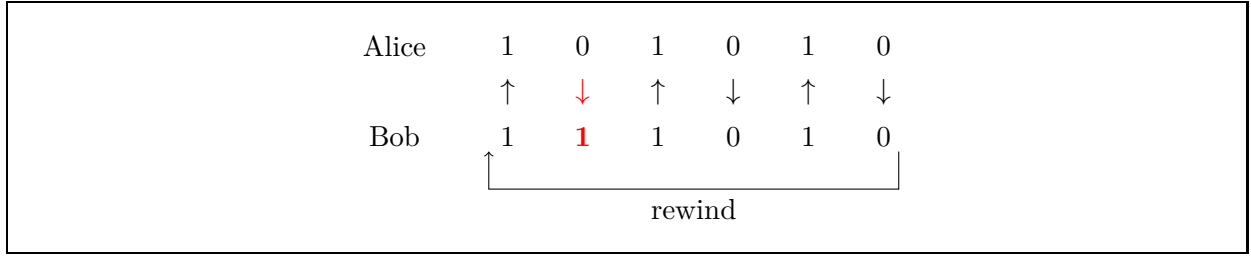

Even more interesting is the fact that the above protocol
is the best possible, in terms of the maximal tolerable noise.
Indeed, we show that in this setting, it is impossible to tolerate noise levels of~$1/6$ or higher.
\begin{theorem}
Any binary protocol with a fixed order of speaking that computes the identity function $f(x,y)=(x,y)$ over a~\bscf, succeeds with probability at most~$1/2$ assuming an error rate of~$1/6$.
\end{theorem}
\begin{proof}
Assume a binary robust protocol $\pi$ that computes the identity function $f(x,y)=(x,y)$ where $x,y$ belong to some domain of size at least~$3$; assume~$|\pi|=N$, and without loss of generality let Alice be the party that speaks at most $T\le N/2$~times in the protocol. We show an adversarial strategy that corrupts at most~$1/3$ of Alice's messages and makes Bob's view look the same for two different inputs of Alice. A similar approach appears in~\cite[Ch.~4]{berlekamp1964block}.

Assume Alice holds one of three inputs, $x_0, x_1, x_2$. 
For a given instance of the protocol, define $x_0[i], x_1[i], x_2[i]$ to be the $i$-th bit sent
by Alice for the respective inputs. Note that the $i$-th transmission may depend on the transcript so far.
If we fix a transcript up to the round where Alice sends her $i$-th bit,
and look at her next transmission, $x_0[i], x_1[i], x_2[i]$, we observe that either the bit value is the same for all $x_0, x_1, x_2$, or it is the same for two of these inputs, and different for the third one. For every $i\le T$, let $\maj(i) = \mathsf{majority}(x_0[i], x_1[i], x_2[i])$, given that previous transmissions are consistent with the adversarial strategy described below, and let~$w[1,\dotsc,T]= \maj(0)\dotsm \maj(T)$.

The adversarial strategy only corrupts Alice, so we should describe what is being sent at each one of the $T$ rounds in which Alice has the right of speak.
The attack consists of two parts: the first $R$ transmissions of Alice, and the last $T-R$ transmissions,  for a number $R$ we set shortly.
For the first part, i.e., any transmission $i \le R$,  
we corrupt the transmission so it equals $\maj(i)$ (i.e., if Alice sends $\maj(i)$ we leave the transmission intact and otherwise we flip the bit). 
The number~$R$ is set to be the minimal round such that for at least one of $x_0,x_1,x_2$, 
the above strategy corrupts exactly $R-2T/3$~bits up to round~$R$ (included). 
It is easy to verify 
we can always find a round $2T/3 \le R \le T$ that satisfies the above:
for every $x_j$ the quantity $d(w[1,\dotsc,R],x_j[1,\dotsc,R])$ starts at $0$, never decreases, and increases at most by one in every round. Furthermore since in the first part at most~$T$ corruptions happen over all rounds and all three inputs, at least for one input $x_j$ the quantity $d(w[1,\dotsc,R],x_j[1,\dotsc,R])$ will grows from $0$ to at most $T/3$. On the other hand, the quantity $R - 2T/3$ increases exactly by one in every round and thus goes from $-2T/3$ to~$T/3$. Therefore, there exists a round $R$ in which $R - 2T/3$ catches up with $d(w[1,\dotsc,R],x_j[1,\dotsc,R])$ for some input~$x_j$. 

Let $x_0$ be the input for which the number of  corruptions up to
round~$R$ is $R-2T/3$; note that 
$x_0$ minimizes $d(w[1,\dotsc,R],x_j[1,\dotsc,R])$, or otherwise one of the other inputs would have satisfied $d(w[1,\dotsc,R'],x_j[1,\dotsc,R'])=R'-2T/3$ for some earlier round $R' < R$. We can therefore assume without loss of generality that,
\begin{equation}\label{eqn:order}
d(w[1,\dotsc,R],x_0[1,\dotsc,R]) \le d(w[1,\dotsc,R],x_1[1,\dotsc,R]) \le d(w[1,\dotsc,R],x_2[1,\dotsc,R]).
\end{equation}
In the second part (the last $T-R$ of Alice's slots), 
we corrupt the $i$-th transmission so it equals~$x_1[i]$.
That is, if Alice holds~$x_1$ we do nothing and if she holds~$x_0$ we flip the bits as needed to correspond to what Alice would have sent on input~$x_1$. 
We do not care about $x_2$ in this second part.

First note that from Bob's point of view, 
the transcripts he sees given that Alice holds~$x_0$ or~$x_1$ are exactly the same.
Next, we claim that for both these inputs, the total amount of corruptions is at most $T/3\le N/6$.
If Alice holds the input~$x_0$, then the total amount of corruptions is at most
\[
d(w[1,\dotsc,R],x_0[1,\dotsc,R])  + (T-R) \le (R-2T/3) + (T-R) = T/3 \le N/6.
\]
If Alice holds~$x_1$ then we do not make any corruption during the last $(T-R)$ rounds, and the total amount of corruptions is at most $d(w[1,\dotsc,R],x_1[1,\dotsc,R])$. Since $w$~is the majority, at each round~$i$, there exists at most a single input $x_j$, for which $d(w[i],x_j[i])=1$,
while for both other inputs $x_{j'}$, the $i$-th transmitted bit is the same as their majority,~$d(w[i],x_{j'}[i])=0$. It follows that 
\[
d(w[1,\dotsc,R],x_0[1,\dotsc,R]) + d(w[1,\dotsc,R],x_1[1,\dotsc,R]) + d(w[1,\dotsc,R],x_2[1,\dotsc,R]) \le R,
\]
thus with Eq.~\eqref{eqn:order}
and the fact that $d(w[1,\dotsc,R],x_0[1,\dotsc,R])=R-2T/3$, we have
\begin{align*}
R-2T/3 + 2d(w[1,\dotsc,R],x_1[1,\dotsc,R]) &\le  R \\
d(w[1,\dotsc,R],x_1[1,\dotsc,R]) & \le T/3 \le N/6.
\end{align*}
\end{proof}

As a corollary of the above, we note that the same impossibility bound of~$1/6$ holds also for the case of standard noisy channel (i.e., without feedback). Clearly, adding the resource of noiseless-feedback can only improve the noise resilience. This simple observation immediately leads to Theorem~\ref{thm:lowerbound1/6}.

\subsection{Protocols with an arbitrary order}
\label{sec:bscf-bin-arbitrary}

When the order of speaking needs not be fixed, 
we can improve the noise resilience of the simulation.
A simple observation is that we can take 
the protocol of Theorem~\ref{thm:bscf-fixed-binary} and change it so
after receiving a $00$ rewind command, the parties rewind only 5 rounds of $\pi$ instead of six. This immediately yields a protocol that resists noise levels up to~$1/5$. 
However, we can do even better. 
We devise a protocol in which
the parties adaptively change the length of the messages
they send and force the adversary to make more corruptions in order to cause the parties to accept a corrupted message. In case the adversary does not corrupt a substantial part of the message, the parties detect the corruption and discard the message. Similar ideas appear in~\cite{AGS13}.

We show a binary protocol that tolerates noise levels of up to~$1/3$, similar to the simulations over feedback channels over large-alphabet.
The bound of~$1/3$ is tight due to the impossibility of Theorem~\ref{thm:imp-bscf-adp} that applies  to binary channels as well.

\begin{theorem}\label{thm:bscf-adp-binary}
For any alternating noiseless binary protocol~$\pi$ of length~$n$, and for any $\eps>0$,
there exists an efficient, deterministic,  robust simulation of~$\pi$ over 
a \bscf with a binary alphabet that takes $O_\eps(n)$~rounds and succeeds assuming a maximal noise rate of~$1/3-\eps$.
\end{theorem}

The idea of the simulation is the following. The parties exchange messages of varying lengths. Each message consists of three parts: 
(\textit{a}) 1 control bit (a \emph{rewind} bit) --- if set, this is an indication that the previous message was corrupted and the protocol should be rewound to the beginning of that message; 
(\textit{b}) 1 bit of information --- the next bit of $\pi$, in case no rewind is due; 
and (\textit{c}) $t\ge1$ confirmation bits, set according to 
how parts (a) and (b) are received at the other side (according to the feedback): 
if the information and rewind bits are received intact, 
the confirmation bits will be `$1$', or otherwise they will be~`$0$'. 
The sender keeps sending confirmation bits, and checking via the feedback what the other side has received, until one of the following happens:
\begin{enumerate}
\item the number of \emph{received} 0-confirmation bits is at least $1/3$ of the length of the current message --- in this case the message is \emph{unconfirmed} and the protocol rewinds to the beginning of that message (so that the sender has the right of speak again to send the same message).
\item the number of received 1-confirmation bits minus the number of received 0-confirmation bits is larger than~$1/\eps$ --- in this case the message is \emph{confirmed} and the parties either rewind the protocol to the previous message of the sender (if the rewind bit is on), or the next simulated bit is the information bit.
\end{enumerate}
The parties perform the above until a total number of $n/\eps^2$ bits are communicated altogether.
We formulate the  simulation protocol in Algorithm~\ref{alg:bscf-binary}. We now continue to prove Theorem~\ref{thm:bscf-adp-binary}.

\begin{algorithm}[htpb]
\caption{Simulation for channels with feedback with a binary alphabet}
\label{alg:bscf-binary}
\begin{algorithmic}[1]
\small

\Statex

Input: an alternating binary protocol $\pi$ of length $n$, a noise parameter $\eps>0$, an input value~$x$.

\Statex

\Statex Initialize $T\gets \emptyset$. 
\Statex 
\Comment  $T=(T^S,T^R,T^F)$ is the simulated transcript, separated to sent, received and feedback bits

\Statex

\For {$i=1,2, \dotsc$  }  
	\If {$\pi(x \mid T^F, T^R)$ is your turn to speak}
		\State $sentLength \gets 2$
		\State $\textit{conf}_0 \gets 0$, $\textit{conf}_1 \gets 0$
		
		\If {$T^S=T^F$}		\Comment No corruptions are known
			\State $\textit{rewind}=0$
		\Else				\Comment The transcript at the other side is corrupt
			\State $\textit{rewind}=1$
		\EndIf
		\Statex 
		
		\State $msg \gets \pi(x \mid T^F,T^R) \circ \textit{rewind}$		
		\State send $msg$
		\While {($\textit{conf}_0< sentLength/3$) and ($\textit{conf}_1-\textit{conf}_0 < 1/\eps$)}  \label{step:confirmation}
			\If {$msg$ received correctly} 			\Comment verify via the feedback
				\State send 1
			\Else
				\State send 0
			\EndIf
			\State  $sentLength\gets sentLength+1$
			\State  $\textit{conf}_b \gets \textit{conf}_b +1$	\Comment  $b$ is the bit received at the other side (learned  via the  feedback) 
		\EndWhile
		\Statex 
		\Statex
		\If {$\textit{conf}_0 \ge sentLength/3$}	\label{step:unconf}	\Comment message is not confirmed
			\State \textbf{continue} (next for loop instance)
		\ElsIf {$\textit{conf}_1- \textit{conf}_0 \ge 1/\eps$}							\label{step:conf}	\Comment message is confirmed: rewind or advance~$T$ \hspace{0.4cm}~ 
		\Statex	\hfill according to info/rewind received at other side
			\If {$\langle\text{rewind bit recorded at the other side}\rangle=0$}
				\State $T^S \gets T^S \circ \pi(x \mid T^S,T^R)$
				\State $T^F \gets T^F \circ \langle\text{info bit recorded at the other side}\rangle$
			\ElsIf {$\langle\text{rewind bit recorded at the other side}\rangle =1$}			
			\Statex \Comment Remove from $T$ the last two simulated rounds

				\State $T^R \gets \mathsf{prefix}_{|T^R|-1}(T^R)$
				\State $T^S \gets  \mathsf{prefix}_{|T^S|-1}(T^S)$
				\State $T^F \gets  \mathsf{prefix}_{|T^F|-1}(T^F)$
			\EndIf
		\EndIf
		\Statex
	\Else			\Comment The other party is the speaker at this round 
		\State Record $msg$, and confirmation bits according to the conditions of the while loop on line~\ref{step:confirmation}.
		\State If $msg$ unconfirmed (line~\ref{step:unconf}), ignore $msg$ and \textbf{continue}.
		\State If $msg$ confirmed (line~\ref{step:conf}):  
		\Statex  \qquad\quad either extend~$T^R$ (if $\textit{rewind}=0$) or delete the suffix bit of~$T^R,T^S,T^F$ (if $\textit{rewind}=1$).
	\EndIf
	
	\Statex
	\State If more than $n/\eps^2$ bits were communicated, terminate and output~$T$
\EndFor

\end{algorithmic}
\end{algorithm}

\begin{proof}
Consider a run of the protocol that did not compute the correct output,
we will show that the noise must have been $\ge1/3-O(\eps)$.
Let $N$ be the amount of times the protocol executed the for-loop, and for 
$i= 1,\dotsc, N$ let $m_i$ be the entire transmission communicated during the $i$-th instance of the loop (i.e., $|m_i|=sentLength$ when reaching line~\ref{step:unconf}).
We know that~$\sum_i |m_i| = n/\eps^2$.

Each message $m_i$ can be confirmed or unconfirmed as explained above. If some message $m_i$ is confirmed it can either be \emph{correct} or \emph{incorrect}
according to whether or not any of its first two bits, $msg$, was flipped.
Divide the messages $m_1, \ldots, m_N$ into three disjoint sets:
\begin{itemize}
\item $U = \{ i \le N\mid \text{$m_i$ is unconfirmed} \}$
\item $C = \{ i \le N \mid \text{$m_i$ is confirmed and correct} \}$
\item $W = \{ i \le N\mid \text{$m_i$ is confirmed and incorrect} \}$
\end{itemize}

It is easy to see that an unconfirmed message has no effect on the simulated transcript as any such message is just ignored.
If a message is confirmed, it can either be interpreted as an information bit or as a rewind request, and the simulation is similar to the algorithm of Theorem~\ref{thm:bscf-adp-large} where each one of the symbols $\{0,1,\back\}$ is encoded into a longer message. 
Specifically, similar to Theorem~\ref{thm:bscf-adp-large}, 
after a single incorrect message  
the simulation takes two correct messages in order
to recover from the error, 
i.e., in order to revert to the state before the corruption.
Also here, multiple erroneous messages just linearly accumulate, hence, 
\begin{claim}\label{clm:correct}
The simulation of a protocol $\pi$ of length~$n$ succeeds as long as
\[
|C|-2|W| \ge n.
\]
\end{claim}

Next, we bound the length and noise rate of a message.
To avoid edge cases caused by rounding, let us assume without loss of generality that $1/\eps$ is an integer (alternatively, replace $\eps$ with $\eps'$ for some $0<\eps'<\eps$ with an integral reciprocal)
\begin{claim}
For any $i$, $|m_i|$ is bounded by $2+ 3/\eps$.
\end{claim}
\begin{proof}
Assume a message reaches length~$2+3/\eps$, and consider its~$3/\eps$ confirmation bits: 
If  $1+1/\eps$ of these bits are zeros, then the message is unconfirmed since $(1+1/\eps)/(2+3/\eps)>1/3$.
Otherwise, there are at most $1/\eps$ zeros and at least $3/\eps-1/\eps \ge 2/\eps$ ones, 
thus the difference between confirmation zeros and ones is at least~$1/\eps$ and the message is confirmed.
\end{proof}

Note that for a confirmed message, $2+1/\eps \le |m_i| \le 2+ 3/\eps$, and for an unconfirmed 
message $3 \le |m_i| \le 2+ 3/\eps$. Since the total amount of bits the protocol communicates is~$n/\eps^2$, we have 
\begin{equation}\label{eqn:sizes}
|C|+|W| < \frac{n}{\eps}.
\end{equation}

The specific length of a message relates to the amount of corruption the adversary must make during that message. 
\begin{itemize}
\item
If $i\in C$ then the first two bits of $m_i$ were received correctly, and the possible noise can only be flipping some of the one-confirmation bits, i.e., the amount of corrupted bits is exactly~$\conf_0$. Since the message was eventually confirmed, it holds that $\conf_1-\conf_0 =  1/\eps$,
and thus $\conf_0 = (|m_i|-2-1/\eps)/2$. 
\item
For unconfirmed messages, $i\in U$, the message gets unconfirmed as soon as $\conf_0 \ge |m_i|/3$. There are two cases: (i) if the information/control bits are correct, then the noise is any 0-confirmation bit, thus $\conf_0  \ge |m_i|/3$; (ii) the information/control bits are corrupt, and then the noise is the corruption of the information/control plus any 1-confirmation bit. 
We have $3\conf_0 \ge |m_i|=(2+\conf_0+\conf_1)$ thus $\conf_1 \ge \frac{2|m_i|}{3}-1\ge \frac{|m_i|}{3}$.
\item
For~$i\in W$, the corruption consists of at least one of the information/control bits and any $\conf_1$ received. We have $\conf_1-\conf_0 \ge1/\eps$ thus $\conf_1 \ge \conf_0 +1/\eps$ or equivalently $\conf_1 \ge (|m_i|-2+1/\eps)/2$.
\end{itemize}

\noindent
Therefore, the global noise rate in any given simulation is lower bounded by
\begin{align*}
\text{Noise Rate}&\ge \frac{ \sum_{i\in C} (|m_i|-2-1/\eps)/2 +  \sum_{i\in U} |m_i|/3 + \sum_{i\in W} (|m_i|+1/\eps)/2}{\sum_{i\in C} |m_i| + \sum_{i\in U} |m_i| + \sum_{i\in W} |m_i|}. 
\end{align*}
We can rewrite the noise rate as
\begin{align*}
& \ge \frac{ \frac12\sum_{i\in C}|m_i|  -\frac12|C|(2+\frac1\eps) +  \frac13\sum_{i\in U} |m_i|  + \frac12\sum_{i\in W}|m_i| +\frac12W\cdot \frac1\eps}{n/\eps^2} \\
& \ge \frac13 +\frac{ \frac16\sum_{i\in C}|m_i|  -\frac12|C|(2+\frac1\eps)  + \frac16\sum_{i\in W}|m_i| +\frac12W\cdot \frac1\eps}{n/\eps^2} \\
& \ge \frac13 + \frac{\frac16\sum_{i\in C}|m_i| - \frac1{2\eps}|C| + \frac16\sum_{i\in W}|m_i| + \frac1{2\eps}|W|}{n/\eps^2} - O(\eps).  \\
\intertext{%
We now use the fact that the simulation 
instance we consider failed to simulate~$\pi$ correctly. 
Using Claim~\ref{clm:correct} we have that $|C| - 2|W| < n$, 
or equivalently, $ |W| > \frac12 (|C|-n)$.
If $|C|<n$ it is trivial that the error rate is $\ge 1/3-O(\eps)$. 
Otherwise, the error rate increases as we increase the length of messages in $C$ and $W$. Since such messages are confirmed, they are of length~$ \ge 1/\eps  $. Then,
}
& \ge \frac13  + \frac{\frac1{6\eps}|C| - \frac1{2\eps}|C| + \frac1{12\eps}(|C|-n) + \frac1{4\eps}(|C|-n)}{n/\eps^2} - O(\eps) \\
& \ge \frac13 - O(\eps). %
\end{align*}
\end{proof}

\section{Coding for Erasure Channels}
In this section we move to discuss \emph{erasure channels}
in which the noise may turn any symbol into an erasure mark~$\bot$. No feedback is assumed in this setting, so the sender is unaware of erased transmissions. 
As in the above sections we seek after the maximal erasure rates that robust  protocols can deal with. In this setting the parties have no longer joint view, and reaching consensus (e.g., regarding who is the next to speak) is not a trivial task. In fact, similar to the case of standard noisy channels~\cite{BR11}, any robust protocol must have a fixed (predetermined) order of speaking.

\begin{theorem}
Any robust protocol over an erasure channel has a fixed order of speaking
\end{theorem}
\begin{proof}
Let $\pi$ be a protocol of length $|\pi|=N$ rounds. 
We denote erasure patterns as strings $E\in\{ \emptyset, \bot\}^N$: if $e_i=\bot$ then the $i$-th transmission is erased, and if $e_i=\emptyset$ it is delivered intact.
Denote with $\owner_{x,y}(R,E) \in \{\text{Alice}, \text{Bob}\}$ the party that owns  round $R\le N$ (i.e., the sender) given the input~$(x,y)$ and erasure pattern~$E\in\{ \emptyset, \bot\}^N$. Note that for any $E_1,E_2$ that agree on the first $R$ indices, $\owner_{x,y}(R,E_1)=\owner_{x,y}(R,E_2)$ i.e., the protocol is causal.

Assume towards contradiction that the theorem does not hold. It easily follows that there must exist some input $(x,y)$, a round $R$ and two different erasure patterns $E_A,E_B$ such that $\owner_{x,y}(R,E_A) \ne \owner_{x,y}(R,E_B)$. Without loss of generality, assume $E_B = \bot^N$. Recall that erasures that occur after round~$R$ cannot affect the owner of round~$R$, hence in the following we only consider patterns in~$\{ \emptyset, \bot\}^R$.

Define a sequence of erasure patterns~$\{E_i\}_{i=0}^R$
that starts with the pattern~$E_A$, and adds at each step an erasure at the $i$-th from last round. 
More precisely,  
 $E_i = E_A \vee \emptyset^{R-i}\bot^i$
for $0\le i \le R$.  Here we treat $\emptyset$ as a $0$ and $\bot$ as a $1$, and  $\vee$ is the equivalent of a bitwise OR function. Note that $E_0=E_A$ and $E_R=E_B$. 
Thus, there must exist some index $0\le i<R$ for which $\owner_{x,y}(R,E_i) \ne \owner_{x,y}(R,E_{i+1})$. Fix this~$i$ for the rest of the proof.

First, we claim that $\owner_{x,y}(R-i,E_i) = \owner_{x,y}(R-i,E_{i+1})$.
This follows since the view of $\owner_{x,y}(R-i-1,E_i)$ up to round $R-i-1$ (including) is the same whether or not its  transmission at the $(R-i)$-th round is erased by the channel.
Then, the identity of party that owns the $(R-i)$-th round must be independent of whether the $(R-i)$-th transmission is erased, i.e., whether the noise is~$E_i$ or~$E_{i+1}$.

Let $p=\owner_{x,y}(R-i,E_i)$. Note that the view of~$p$ up to round~$R$ is the same for both $E_i$ and $E_{i+1}$, since up to round $R-i-1$ the erasure pattern (and thus, the transcript) is exactly the same, round~$R-i$ is a message sent by~$p$ and after round $R-i$ all the messages are erased, thus $p$ cannot learn whether his message at the $(R-i)$-th round was erased or not. Hence, party~$p$ cannot distinguish the two cases, and it will error on the owner of round $R$, for either $E_i$ or $E_{i+1}$.
\end{proof}

It is already known that $1/2$ is a tight bound on the erasure rate 
robust interactive protocols can tolerate over erasure channels~\cite{FGOS13}.
Specifically,  \cite{FGOS13}~shows that no protocol can resist an erasure level of~$1/2$, due to the trivial attack that
completely erases one party.
Furthermore, they show that for any $\eps>0$, the coding scheme of~\cite{BR11} can tolerate an erasure level of~$1/2-\eps$. 
However, that coding scheme has several drawbacks.
First, it takes exponential time due 
assuming \emph{tree codes}, a data structure whose efficient construction is still unknown (see~\cite{schulman96,GMS11,Braverman12,MS14}). 
Furthermore, as~$\eps\to0$ and the erasure level approaches~$1/2$, the tree code needs to be more powerful, which implies the increase of the alphabet size (as a function of~$\eps$).

In the following subsection~\ref{sec:large-erasure} we provide a simple, 
computationally efficient coding scheme
for interactive communication over erasure channels
in which the alphabet is small (namely, $6$-ary),
and yet it tolerates erasure rates up to~$1/2-\eps$. 
Then, in subsection~\ref{sec:bin-erasure} we transform this protocol to obtain the best known protocol for binary erasure channels, resisting an erasure level of up to~$1/3-\eps$. It is there where having a fixed, relatively small, alphabet leads to an improved noise tolerance.  

\subsection{Erasure channels with a ``large'' alphabet}
\label{sec:large-erasure}
\begin{theorem}\label{thm:bec-fixed-large}
For any alternating noiseless binary protocol $\pi$ of length~$n$, and for any~$\eps>0$,
there exists an efficient, deterministic,  robust simulation of~$\pi$ over 
an \bec with a $6$-ary alphabet, that takes $O_\eps(n)$~rounds and succeeds assuming a maximal erasure rate of~$1/2-\eps$.
\end{theorem}

The main idea is the following.
The parties talk in alternating rounds, in each of which they send a symbol
$m\in \mathsf{Info}\times\mathsf{Parity}$
where $\mathsf{Info} = \{0,1\}$ is the next information bit according to~$\pi$ (given the accepted simulated transcript so far)
and $\mathsf{Parity}= \{0,1,2\}$ is the parity of the round in~$\pi$ being simulated, modulus~$3$.

Assume $T$ is the transcript recorded so far, and let~$p={|T| \mod 3}$.
If  a received $m$ has parity~$p+1$, the receiving party accepts this message and extends~$T$ by one bit according to the $\mathsf{Info}$ part.  Otherwise, or in the case $m$~is erased, the party ignores the message, and resends its last sent message.

Since messages might get erased, the parties might get \emph{out-of-sync}, e.g.\@ when one party extends its accepted~$T$ while the other party does not. 
However, this discrepancy is limited to a single bit, 
that is, the length of Alice's~$T$ differs from Bob's by at most $\pm1$. 
Sending a parity---the length of the current $T$ modulus~$3$---gives full information on the status of the other side, 
and allows the parties to regain synchronization.
We formalize the above as Algorithm~\ref{alg:bec-large}.

\begin{algorithm}[!ht]
\caption{Simulation for erasure channels}
\label{alg:bec-large}
\begin{algorithmic}[1]
\small

\Statex
Input: an alternating binary protocol $\pi$ of length $n$, 
s noise parameter $\eps>0$, an input value~$x$.

\Statex\Statex
Assume a fixed alternating order of speaking: Alice is the sender on odd $i$'s,  and Bob is the sender on even $i$'s.

\Statex

\State
Initialize: $T\gets \emptyset$, $p\gets 0$, and $m \gets (0,0)$. Set $N= \lceil n/\eps \rceil$.

\Statex

\For {$i = 1$ to $N$}					%
	\Statex
	\If {Sender}						%
		
		\If {your turn to speak according to~$\pi(\cdot \mid T)$}
			\State $t_{send} \gets \pi(x \mid T)$
			\State $m \gets  (t_{send},  \ (p+1)\!\!\mod 3)$
			\State $T \gets T \circ t_{send}$
			\State $p \gets |T| \mod 3$
			\State send $m$
		\Else
		\State send the $m$ stored in memory
		\EndIf
	\EndIf
	\Statex
	\If {Receiver}
		\State	record $m' = (t_{rec}, \ p')$
		\If {$m'$ contains no erasures and $p' \equiv p+1 \mod 3$}
			\State $T \gets T \circ t_{rec}$		\label{step:rec}
		\EndIf
	\EndIf
\EndFor

\Statex
\State Output~$T$
\end{algorithmic}
\end{algorithm}

\begin{proof}(Theorem~\ref{thm:bec-fixed-large}.)
First we set some notations for this proof.
For a variable $v$ denote with $v_A$ Alice's instance of the variable (resp.~$v_B$ for Bob's instance) and with $v(i)$ the state of the variable at the beginning of the $i$-th instance of the for-loop.
For a string $a=a[0]\cdots a[k]$ denote by $\Trim(a)=a[0]\cdots a[k-1]$ the string obtained by trimming the last index, that is $\mathsf{prefix}_{k-1}(a)$.

\goodbreak
Next, we analyze Algorithm~\ref{alg:bec-large} and show it satisfies the theorem.
\begin{claim}\label{clm:discrepancy}
For any $i \le N$, 
$\big| |T_A(i)| - |T_B(i)| \big | \le 1$.
\end{claim}
\begin{proof}
We prove by induction on the round~$i$. The base case~$i=1$ trivially holds.
Now assume the claim holds up to the beginning of the $i$-th iteration of the loop. We show that the claim still holds at the end of the $i$-th iteration. 

We consider several cases.
(\textit{a}) $|T_A(i)| = |T_B(i)|$. Then trivially each of $T_A,T_B$ can extend by at most a single bit and the claim still holds.
(\textit{b}) $|T_A(i)| = |T_B(i)|+1$. 
Note that this situation is only possible if the $|T_A(i)|$-th round of $\pi$ is Alice's round:
otherwise, in a previous round $T_A$ was of length $|T_B(i)|$ and it increased by one bit at line~\ref{step:rec}. But for this to happen, it must be that the received parity was $p'=|T_B(i)|+1$. Yet, $T_B$ never decreases, so such a parity could be sent by Bob only if at that same round, $|T_B| = |T_A|+1 -3k$  for some $k\in\mathbb{N}$. But then $|T_B| \le |T_A|-2$, which contradicts the induction hypothesis.

Consider round~$i$. 
If $i$ is odd, Alice just resends~$m_A(i)$ from her memory, since according to~$T_A(i)$ the next bit to simulate belongs to Bob; she thus doesn't change~$T_A$. 
Bob might increase~$T_B(i)$ since~$p_A = p_B +1$.
So the claim still holds for an odd~$i$.
If $i$ is even, Bob is resending~$m_B(i)$ since according to~$T_B(i)$ the next bit to simulate belongs to Alice. Yet Alice will not increase her $T_A$ since the parity mismatches. There is no change in $T_A,T_B$ in this case and thus the claim holds.
The third case (\textit{c}) $|T_B(i)| = |T_A(i)|+1$, is symmetric to~(\textit{b}).
\end{proof}

\begin{claim}
For any $i \le N$, 
$T_A$ and $T_B$ hold a correct prefix of the transcript of~$\pi(x,y)$.
\end{claim}
\begin{proof}
Again, we prove by induction on the round~$i$. 
The claim trivially holds for $i=1$. 
Now, assume that at some round $i$, $T_A(i)$ and $T_B(i)$ are correct. 
We show they are still correct at round~$i+1$.

Consider round~$i$, and assume without loss of generality, that Alice is the sender.
Then, if Alice adds a bit to~$T_A$, this bit is generated 
according to~$\pi$ given a \emph{correct} prefix~$T_A$, which means that the generated bit is correct. 
As for the receiver, note that any message that contains erasures is being ignored. 
Thus, the only possibility for $T_B$ to be incorrect is if the parties are out of sync, and Bob receives a message that does not correspond to the round of $\pi$ that Bob is simulating.  
However, if Bob accepts the received bit it must be that the received parity satisfies
$p'=|T_B(i)|+1 \mod 3$. 

Since $\big| |T_A(i)| - |T_B(i)| \big | \le 1$ (claim~\ref{clm:discrepancy}), 
there are two cases here. Either $|T_A|=|T_B|$ which implies that Alice simulated the same round as Bob, and the received bit does correspond to the round simulated by Bob;
or $|T_A(i)| \ne |T_B(i)|$, so that Alice sent a message $m$ that was stored in her memory, in which $p'=|T_B(i)|+1 \mod 3$. 
It is easy to verify that at any given round, the message saved in the memory
is the one that was generated given the transcript $\Trim(T_A)$  
(otherwise, a new $m$ must have been generated by Alice).
Along with the constraint on the parity, it must be the case that $|T_A(i)| = |T_B(i)|+1$
which means the stored $m$ is indeed the correct bit expected by Bob, and the claim holds in this case as well.
\end{proof}

After establishing the above properties of the protocol, 
we consider how the protocol advances at each round. 
We begin by exploring rounds in which 
the transmission is not erased.
\begin{lemma}\label{lem:no-erasure}
Assume that no erasure happens during the $i$-th transmission.
Then, \\
(a) if  $|T_A(i)| = |T_B(i)|$, then $|T_A(i+1)| = |T_A(i)|+1 = |T_B(i+1)|$.\\
(b) if $i$ is even and $|T_A(i)| < |T_B(i)|$, then $|T_A(i+1)| = |T_A(i)|+1$. 
\end{lemma}
\begin{proof}
Part (\textit{a}): trivial from Algorithm~\ref{alg:bec-large}.
Part (\textit{b}): recall that this situation is only possible if $|T_B(i)| = |T_A(i)| +1$ and that the $|T_B(i)|$-th round in~$\pi$ belongs to Bob (see the proof of claim~\ref{clm:discrepancy} above). Thus, if $i$ is even, Bob is the sender, and since $\pi(y \mid T_B(i))$ is Alice's round, Bob will retransmit the message $m_B(i)$ stored in his memory; this corresponds to running $\pi$ 
given the transcript $\Trim(T_B(i))=T_A(i)$. The parity in $m_B(i)$ is thus $p_B(i)=p_A(i)+1$ and Alice accepts this transmission and extends $T_A$. 
\end{proof}
\noindent A symmetric claim  holds for Bob, replacing even and odd rounds, and assuming $|T_B(i)|\le |T_A(i)|$.
Next, consider rounds in which the transmission is erased. 
These 
can only cause the parties a discrepancy of a single bit in their (accepted) simulated transcripts.
\begin{lemma}\label{lem:yes-erasure}
Assume an erasure happens during the $i$-th transmission.\\
If  $|T_A(i)| = |T_B(i)|$ then if $i$ is odd, $|T_A(i+1)| = |T_A(i)|+1$ while $|T_B(i+1)| = |T_B(i)|$, 
and if $i$ is even, $|T_B(i+1)| = |T_B(i)|+1$ while $|T_A(i+1)| = |T_A(i)|$.
Otherwise (i.e., $|T_A(i)| \ne |T_B(i)|$), there is no change in~$T_A, T_B$.
\end{lemma}
\begin{proof}
For the first part of the lemma, assume $j= |T_A(i)|=|T_B(i)|$.
Let~$R$ be the first round for which $|T_A(R)| = |T_B(R)|=j$,
and assume that at the previous round,  
$|T_A(R-1)| < |T_B(R-1)|$. 
Since the transcripts are of equal length at round~$R$,   Alice must have been the receiver at the $R-1$ round, which means that the $(j-1)$-th round of~$\pi$ belongs to Bob. Therefore, the $j$-th round of~$\pi$ belongs to Alice, who is also the sender of round~$R$ in the simulation (i.e., $R$ is odd), thus she extends her transcript in one bit at this round as well. Note that $T$ never decreases, thus we must have that~$R=i$, and the claim holds for this case. Also note that a similar reasoning applies for the edge case of~$i=1$.
If instead we have $|T_A(R-1)| > |T_B(R-1)|$, a symmetric argument implies that Bob owns the $j$-th round of~$\pi$ and is also the sender in the $R$-th round ($R=i$ is even), thus the claim holds for this case as well.

Next, for the second part of the claim, assume that $|T_A(i)| \ne |T_B(i)|$, and without loss of generality 
assume that~$|T_A(i)| < |T_B(i)|$. We know that the gap between the two transcripts is at most~$1$, thus if there is any change in these transcripts during the $i$-th rounds, either both parties increase their transcript, or only Alice does. Since the receiver sees an erasure, that party ignores the incoming message and doesn't change his transcript, so it cannot be that both parties extend their transcripts.
On the other hand, if $|T_B(i)| = |T_A(i)|+1$ then the $|T_B(i)|$-th round of~$\pi$ belongs to Bob (see the proof of Claim~\ref{clm:discrepancy}), and Alice will increase $T_A$ only if she is the receiver. However, in this case Alice sees an erasure and ignores the message, keeping her transcript unchanged. 
\end{proof}

Lemma~\ref{lem:no-erasure} and Lemma~\ref{lem:yes-erasure} lead to the following corollary.
\begin{corollary}
Any erasure causes at most two rounds in which neither $T_A$ nor $T_B$ extend.
\end{corollary}
\begin{proof}
By Lemma~\ref{lem:yes-erasure}, the only situation where there is no change in both $T_A,T_B$ at some round~$i$ where the transmission is erased, is when $|T_A(i)| \ne |T_B(i)|$. Then, in the next round $i+1$ (in which the transcription is not erased), the transcripts will catch up if the sender is the party that holds the longer transcript. If this is not the case and the sender is the one with the shorter transcript, then in round~$i+2$ the sender is the one with the longer transcript, and the transcript must catch up by Lemma~\ref{lem:no-erasure}. In all other cases, at least one of the transcripts increases during round $i$ or $i+1$.
\end{proof}
Thus, if the adversary is limited to making $(1/2-\eps)N = N/2-n$ erasures, 
it must hold that $|T_A(N)|+|T_B(N)| \ge N - 2\cdot (N/2-n) = 2n$
and since $\big ||T_A(N)|-|T_B(N)|\big| \le 1$, each of them must be of length at least~$n$.
This concludes the proof of the Theorem~\ref{thm:bec-fixed-large}.
\end{proof}

\subsection{Erasure channels with a binary alphabet}
\label{sec:bin-erasure}
By encoding each symbol from our $6$-ary alphabet using a binary code containing at least six codewords with maximal relative distance $\delta_6$ we can immediately get a binary scheme that resists erasure rates of up to~$\delta_6/2-\eps$. To our knowledge the maximal achievable distance $\delta_6$  is $6/10$ (see~\cite{BBMOS78}). This leads to the following corollary:

\begin{corollary}
For any alternating noiseless binary protocol $\pi$ of length~$n$, and for any $\eps>0$,
there exists an efficient, deterministic,  robust simulation of~$\pi$ over 
a binary~\bec, that takes $O_\eps(n)$~rounds and succeeds assuming a maximal noise rate of at most~$3/10-\eps$.
\end{corollary}

Finally, we use the above ideas to devise a binary protocol that resists 
an adversarial erasure rate of up to~$1/3-\eps$. 
The idea is to reduce the number of different messages used by the underlying simulation:
since Algorithm~\ref{alg:bec-large} assumes a $6$-ary alphabet, 
the best code has maximal distance $\delta_6 = 6/10$ which, as mentioned, leads to a maximal resilience
of ${\delta_6/2-\eps} = 3/10 - \eps$. However, were the alphabet in use smaller, say $4$-ary, then we could have used better codes with a higher relative distance $\delta_4 = 2/3$ and achieve a maximal
resilience of $\delta_4/2-\eps =  1/3-\eps$. 

In the following we adapt Algorithm~\ref{alg:bec-large} to use an alphabet size of at most~$4$, and obtain the following.
\begin{theorem}
For any alternating noiseless binary protocol $\pi$ of length~$n$, and for any $\eps>0$,
there exists an efficient, deterministic,  robust simulation of~$\pi$ over 
a binary~\bec, that takes $O_\eps(n)$~rounds and succeeds assuming a maximal noise rate of at most~$1/3-\eps$.
\end{theorem}

\begin{proof}
Each message in Algorithm~\ref{alg:bec-large} consists of two parts: an information bit, and a parity modulus three. In order to reduce the number of possible messages we introduce a simple preprocessing step that takes an alternating protocol $\pi$ of length~$n$ 
and converts it into a protocol $\pi'$ of length $3n$ by padding each two consecutive transmutations of $\pi$ with two vacuous transmissions (say, of the value~$1$). 
That is, if the communication in $\pi$ is  the bitstring $a_1,b_1,a_2,b_2\dotsc$, then in $\pi'$ the parties communicate $a_1,1,1,b_1,1,1,a_2,1,1,\dotsc$ (recall that the protocol is alternating, thus Alice sends the odd bits, and Bob the even ones). 

In the preprocessed $\pi'$, both parties know that a bit of information lies only in transmissions whose parity is~$0$ (mod~3). Thus, for the other parities there is no need to send the information bit --- it is always~$1$! This reduces the size of the alphabet in use, specifically, 
the parties send messages
out of the following message space\footnote{In Algorithm~\ref{alg:bec-large} the first information bit will actually have parity~1 rather than 0; we can alter ${\cal M}$ to have the information bit on parity~1, and the rest remains the same.} 
$${\cal M} = \left\{ 0\times (\textrm{mod}~0)\ ,\  1\times (\textrm{mod}~0)\ ,\ 1\times (\textrm{mod}~1)\ ,\ 1\times (\textrm{mod}~2)\right\}.$$

Now that the message space is of size~$4$ we can encode each message using a binary code of relative distance $\delta_4 = 2/3$. For instance, we can use the code $\{000,011,110,101\}$ (cf.~\cite{BBMOS78}).
Similar to Algorithm~\ref{alg:bec-large}, the obtained simulation is deterministic, efficient and takes $O_\eps(n)$ rounds. As for its noise resilience,
for any $\eps > 0$ the underlying Algorithm~\ref{alg:bec-large}
can be set to resist up to $1/2-\tfrac32\eps$ erased messages (Theorem~\ref{thm:bec-fixed-large}). 
Since each message is coded into a binary string, 
in order to erase a codeword, $2/3$ of its bits must be erased. 
Therefore, in the concatenated algorithm we can resist  
a maximal erasure rate of~$2/3 \cdot (1/2-\tfrac32\eps) = 1/3 - \eps$.
\end{proof}

\small
\newcommand{\etalchar}[1]{$^{#1}$}

\normalsize

\appendix

\end{document}